\documentclass[a4paper]{article}

\usepackage[utf8]{inputenc}
\usepackage[english]{babel}

\usepackage{fullpage} 

\usepackage[textsize=scriptsize]{todonotes}

\usepackage{authblk}

\usepackage[shortlabels]{enumitem}

\usepackage{amsmath}
\usepackage{amssymb}
\usepackage{bbm} 
\usepackage{mathrsfs}

\usepackage{url}
\usepackage{hyperref}
\usepackage{bookmark}
\hypersetup{
    colorlinks = true,
    citecolor = blue,
    linkcolor = blue,
    bookmarksnumbered
}

\usepackage{graphicx}
\usepackage[width=.9\textwidth]{caption}
\usepackage{subcaption}
\usepackage{float}

\usepackage{tikz}
\usetikzlibrary{calc,fit,shapes}

\usepackage{cite}

\usepackage[ruled, vlined, linesnumbered]{algorithm2e}
\SetKwFor{RepTimes}{repeat}{times}

\usepackage[noabbrev]{cleveref}

\usepackage{xspace}

\usepackage{centernot}

\usepackage[misc]{ifsym}

\usepackage{amsthm}
\theoremstyle{plain}
    \newtheorem{theorem}{Theorem}
    \newtheorem{lemma}{Lemma}
    
    \newtheorem{claim}{Claim}
    
\theoremstyle{definition}
    \newtheorem{definition}{Definition}
    \newtheorem{remark}{Remark}
\newtheoremstyle{named}{}{}{}{}{\bfseries}{.}{.5em}{#3}
\theoremstyle{named}
    \newtheorem*{namedthm}{}

\crefname{definition}{definition}{definitions}
\crefname{theorem}{theorem}{theorems}
\crefname{corollary}{corollary}{corollaries}
\crefname{lemma}{lemma}{lemmas}
\crefname{prop}{proposition}{propositions}
\crefname{claim}{claim}{claims}
\crefname{remark}{remark}{remarks}
\newcommand{\R}{\mathbb{R}}
\newcommand{\Z}{\mathbb{Z}}

\newcommand{\bs}[1]{\boldsymbol{#1}}
\newcommand{\1}{\bs{1}}

\newcommand{\Par}[1]{\left( #1 \right)}

\newcommand{\Abs}[1]{\left| #1 \right|}
\newcommand{\CBra}[1]{\left\{ #1 \right\}}

\newcommand{\graph}{G}
\newcommand{\vertexSet}{V}
\newcommand{\edgeSet}{E}
\newcommand{\graphVE}{\graph=(\vertexSet,\edgeSet)}

\newcommand{\shortEdge}[2]{#1#2}
\newcommand{\symdif}{\triangle}
\newcommand{\matching}{M}

\newcommand{\capacity}{c}
\newcommand{\weight}{w}
\newcommand{\graphwc}{(\graph,\weight,\capacity)}

\newcommand{\gwcm}{[\graphwc,\matching]}
\newcommand{\gwcmaux}{[(\graph',\weight',\1),\matching']}
\newcommand{\tb}{\text{tb}} 

\newcommand{\nuG}{\nu(\graph)}
\newcommand{\nufG}{\nu_f(\graph)}
\newcommand{\nufcG}{\nu_f^\capacity(\graph)}
\newcommand{\nucG}{\nu^\capacity(\graph)}

\newcommand{\taufcG}{\tau_f^\capacity(\graph)}

\newcommand{\Po}{\text{P}\xspace}
\newcommand{\NP}{\text{NP}\xspace}

\def\keywordname{{\bf Keywords:}}
\providecommand{\keywords}[1]{\def\and{{\textperiodcentered} }
\par\addvspace\baselineskip
\noindent\keywordname\enspace\ignorespaces#1}

\begin{document}

\title{Stabilization of Capacitated Matching Games}
\author[1]{Matthew Gerstbrein}
\author[2]{Laura Sanità}
\author[3]{Lucy Verberk\(^{(\text{\Letter})}\)}
\affil[1]{
    University of Waterloo, Canada. \protect \\
    {\normalsize\tt{mlgerstbrein@uwaterloo.ca}}
}
\affil[2]{
    Bocconi University of Milan, Italy. \protect \\
    {\normalsize\tt{laura.sanita@unibocconi.it}}
}
\affil[3]{
    Eindhoven University of Technology, Netherlands. \protect \\
    {\normalsize\tt{l.p.a.verberk@tue.nl}}
}
\date{}
\maketitle

\begin{abstract}
    An edge-weighted, vertex-capacitated graph \(\graph\) is called \emph{stable} if the value of a maximum-weight capacity-matching equals the value of a maximum-weight \emph{fractional} capacity-matching. Stable graphs play a key role in characterizing the existence of stable solutions for popular combinatorial games that involve the structure of matchings in graphs, such as network bargaining games and cooperative matching games. 
    
    The vertex-stabilizer problem asks to compute a minimum number of players to block (i.e., vertices of \(\graph\) to remove) in order to ensure stability for such games. The problem has been shown to be solvable in polynomial-time, for unit-capacity graphs. This stays true also if we impose the restriction that the set of players to block must not intersect with a given specified maximum matching of \(\graph\).
    
    In this work, we investigate these algorithmic problems in the more general setting of arbitrary capacities. We show that the vertex-stabilizer problem with the additional restriction of avoiding a given maximum matching remains polynomial-time solvable. Differently, without this restriction, the vertex-stabilizer problem becomes NP-hard and even hard to approximate, in contrast to the unit-capacity case. 
    
    Finally, in unit-capacity graphs there is an equivalence between the stability of a graph, existence of a stable solution for network bargaining games, and existence of a stable solution for cooperative matching games. We show that this equivalence does not extend to the capacitated case.
    
\end{abstract}

\section{Introduction}

\emph{Network Bargaining Games} (NBG) and \emph{Cooperative Matching Games}  (CMG) are popular combinatorial games involving the structure of matchings in graphs. CMG were introduced in the seminal paper of Shapley and Shubik 50 years ago~\cite{Shapley1971Assignment}, and have been widely studied since then. NBG are relatively more recent, and were defined by Kleinberg and Tardos~\cite{Kleinberg2008Balanced} as a generalization of Nash's 2-player bargaining solution~\cite{Nash1950}.

Instances of these games are described by a graph \(\graphVE\) with weights \(w\in\R_{\geq0}^\edgeSet\), where the vertices and the edges model the players and their potential interactions, respectively. The value of a \emph{maximum-weight matching}, denoted as \(\nuG\), is the total value that players can collectively accumulate. The goal, roughly speaking, is to assign values to players in such a way that players have no incentive to deviate from the current allocation. 

Formally, in an instance of a NBG, players want to enter in a deal with one of their neighbours, and agree on how to split the value of the deal given by the weight of the corresponding edge. Hence, an outcome is naturally associated with a matching \(\matching\) of \(\graph\) representing the deals, and allocation vector \(y\in\R_{\geq0}^\vertexSet\) with \(y_{u}+y_{v}=\weight_{uv}\) if \(\shortEdge{u}{v}\in\matching\), and \(y_v=0\) if \(v\) is not matched. An outcome \((\matching,y)\) is \emph{stable} if each player's allocation \(y_u\) is at least as large as their \emph{outside option}, formally defined as \(\max_{v: uv \in \edgeSet\setminus\matching} \CBra{\weight_{uv} - y_v}\).

In an instance of a CMG, one wants to find an allocation of total value \(\nuG\), given by a vector \(y\in\R_{\geq0}^\vertexSet\) in which no subset of players can gain more by forming a coalition. This condition is enforced by the constraint \(\sum_{v\in S} y_v \geq \nu(\graph[S])\) for all \(S\subseteq \vertexSet\), where \(\graph[S]\) indicates the subgraph of \(\graph\) induced by the vertices in \(S\). Such an allocation is called \emph{stable}, and the set of stable allocations constitutes the \emph{core} of the game.

Despite having been defined in different contexts, there is a tight link between stable solutions of these types of games. In particular, if each game is played on the same graph \(\graph\), then it has been shown that either a stable solution exists for both games, or for neither game. This follows as both games admit the same polyhedral characterization of instances with stable solutions~\cite{Deng1999Algorithmic,Kleinberg2008Balanced}. Specifically, a stable solution exists if and only if \(\nuG\) equals the value of the standard linear programming (LP) relaxation of the maximum matching problem defined as
\begin{equation}
    \nufG := \max\CBra{ \weight^\top x : \sum_{u: uv \in E} x_{uv}  \leq1 \; \forall v\in\vertexSet, \; x\geq0 }.
\end{equation}

A graph \(\graph\) for which \(\nuG=\nufG\) is called \emph{stable}. As a result of this characterization, it is easy to see that there are graphs which do not admit stable solutions (to either type of game), such as odd cycles. Given that not all graphs are stable, naturally arises the \emph{stabilization} problem of how to minimally modify a graph to turn it into a stable one. Stabilization problems attracted a lot of attention in the literature in the past years (see e.g.~\cite{Biro10,Konemann12,Bock2015Finding,CHANDRASEKARAN201956,Ahmadian2018Stabilizing,Chandrasekaran2017,Ito2017Efficient,Koh2020Stabilizing}). 

In this context, very natural operations to stabilize graphs are edge- and vertex-removal operations. Those have an interesting interpretation: they correspond to blocking interactions and players, respectively, in order to ensure a stable outcome. While removing a minimum number of edges to stabilize a graph is NP-hard already for unit weight graphs~\cite{Bock2015Finding}, and even hard-to-approximate with a constant factor~\cite{G18,Koh2020Stabilizing}, stabilizing the graph via vertex-removal operations turned out to be solvable in polynomial-time. Specifically, \cite{Ahmadian2018Stabilizing} and \cite{Ito2017Efficient} showed that computing a minimum-cardinality set of players to block in order to stabilize an unweighted graph (called the \emph{vertex-stabilizer problem}) can be done in polynomial time. Furthermore, \cite{Ahmadian2018Stabilizing} showed that computing a minimum set of players to block in order to make a given maximum matching realizable as a stable outcome (called the \emph{\(\matching\)-vertex-stabilizer problem}) is also efficiently solvable. The authors of \cite{Koh2020Stabilizing} showed that both results generalize to weighted graphs.

This paper focuses on \emph{Capacitated NBG}, introduced by Bateni et al~\cite{Bateni2010Cooperative} as a generalization of NBG, to capture the more realistic scenario where players are allowed to enter in more than one deal. This generalization can be modeled by allowing for vertex capacities \(\capacity\in\Z_{\geq0}^\vertexSet\). The notion of a matching is therefore generalized to a \emph{\(\capacity\)-matching}, where each vertex \(v\in\vertexSet\) is matched with at most \(\capacity_v\) vertices. In this case, the value of a maximum-weight \(\capacity\)-matching of a graph \(\graph\) is denoted as \(\nucG\), and the standard LP relaxation is given by
\begin{equation}\label{eq:fracmatching}
    \nufcG := \max \CBra{ \weight^\top x : \sum_{u: uv \in E} x_{uv}\leq\capacity_v \; \forall v\in\vertexSet, \; 0\leq x\leq1 }.
\end{equation}

Similarly to the unit-capacity case, an outcome to the NBG is associated with a \(\capacity\)-matching \(\matching\) and a vector \(a\in\R_{\geq0}^{2\edgeSet}\) that satisfies \(a_{uv}+a_{vu}=\weight_{uv}\) if \(\shortEdge{u}{v}\in\matching\), and \(a_{uv}=a_{vu}=0\) otherwise. The concepts of outside option and stable outcome can be defined similarly as in the unit-capacity case, see \cite{Bateni2010Cooperative}. 

The authors of \cite{Bateni2010Cooperative} proved that the LP characterization of stable solutions generalize, i.e., there exist a stable outcome for the capacitated NBG on \(\graph\) if and only if \(\nucG=\nufcG\) (i.e., \(\graph\) is \emph{stable}). Farczadi et al~\cite{Farczadi2013Network} show that some other important properties of NBG extend to this capacitated generalization, such as the possibility to efficiently compute a so-called \emph{balanced} solution (we refer to \cite{Farczadi2013Network} for details). 

The goal of this paper is to investigate whether the other two significant features of NBG mentioned before generalize to the capacitated setting. Namely:
\emph{
\begin{itemize}
    \item[(i)] Can one still efficiently stabilize instances via vertex-removal operations?
    \item[(ii)] Does the equivalence between existence of stable allocations for capacitated CMG and existence of stable solutions for capacitated NBG still hold?
\end{itemize}
}

\paragraph{Our Results.}
    In this paper we provide an answer to the above questions. 
    
    We investigate the \(\matching\)-vertex-stabilizer problem and the vertex-stabilizer problem in the capacitated setting
    in \cref{sec: M-vertex-stab,sec: vertex-stab}, respectively. While for unit-capacity graphs both problems are efficiently solvable, we show that adding capacities makes the complexity status of the vertex-stabilizer problem diverge. In particular, we prove that the vertex-stabilizer problem is \NP-complete, and no \(n^{1-\varepsilon}\)-approximation is possible, for any \(\varepsilon >0\), unless P=NP. Note that a trivial \(n\)-approximation algorithm can be easily developed.
    
    In contrast, we show that the \(\matching\)-vertex-stabilizer problem is still polynomial-time solvable in the capacitated setting. Our results here extend those of \cite{Koh2020Stabilizing} for unit-capacity graphs, and builds upon an auxiliary construction of \cite{Farczadi2013Network}. 

    Finally, in \cref{sec: CMG} we show that the equivalence between stability of a graph, existence of a stable allocation for CMG and existence of a stable outcome for NBG does \emph{not} extend in the capacitated setting. In particular, we provide an unstable graph which does attain a stable allocation for the capacitated CMG\footnote{It is stated in \cite{Farczadi2015Matchings} (theorem 2.3.9) that a stable allocation for capacitated CMG exists iff \(\graph\) is stable, but our example shows this statement is not correct.}.
\section{Preliminaries and Notation}\label{sec: preliminaries}

\paragraph{Problem definition.} 
    A set \(S \subseteq V\) is called a {\bf vertex-stabilizer} if \(\graph\setminus S\) is stable, where \(\graph\setminus S\) is the subgraph induced by the vertices \(\vertexSet\setminus S\). We say that a vertex-stabilizer \(S\) \emph{preserves} a matching \(\matching\) of \(\graph\) if \(\matching\) is a matching in \(\graph\setminus S\). 

    We now formally define the stabilization problems considered in this paper.

    \smallskip
    \noindent
    \textbf{Vertex-stabilizer problem:} given \(\graphVE\) with edge weights \(\weight\in\R_{\geq0}^\edgeSet\) and vertex capacities \(\capacity \in \Z^\vertexSet_{\geq 0}\), find a vertex-stabilizer of minimum cardinality. 
    
    \smallskip
    \noindent
    \textbf{\(\matching\)-vertex-stabilizer problem:} given  \(\graphVE\) with edge weights \(\weight\in\R_{\geq0}^\edgeSet\), vertex capacities \(\capacity \in \Z^\vertexSet_{\geq 0}\), and a maximum-weight \(\capacity\)-matching \(\matching\), find a vertex-stabilizer of minimum cardinality among the ones preserving \(\matching\).
    
    \smallskip
    An instance of the vertex-stabilizer problem is stable if \(\graph\) is stable. 
    An instance of the \(\matching\)-vertex-stabilizer problem is stable if \(\graph\) is stable, and \(\matching\) is a maximum-weight \(\capacity\)-matching in \(\graph\). Without loss of generality, we can assume that \(\capacity_v\) is bounded by the degree of \(v\in\vertexSet\).
    
\paragraph{Notation.}
    We refer to a graph with edge weights and vertex capacities as \(\graphwc\), and to a graph with a \(\capacity\)-matching \(\matching\) as \(\gwcm\).
    For a vertex \(v\), we let \(\delta(v)\) be the set of edges of \(\graph\) incident into it, we let \(N(v)\) be the set of its adjacent neighbours, and \(N^+(v) = N(v) \cup \{v\}\). For \(F\subseteq\edgeSet\), we denote by \(d_v^F\) the degree of \(v\) in \(\graph\) with respect to the edges in \(F\). We define \(\weight(F):=\sum_{e\in F} \weight_e\). 
    Given a \(\capacity\)-matching \(\matching\), we say that \(v\in\vertexSet\) is \emph{exposed} if \(d_v^\matching=0\), \emph{covered} if \(d_v^\matching>0\), \emph{unsaturated} if \(d_v^\matching<c_v\) and \emph{saturated} if \(d_v^\matching=c_v\). We also use these terms for feasible solutions \(x\) of \eqref{eq:fracmatching}, called \emph{fractional \(\capacity\)-matchings}, e.g., \(v\) is exposed if \(\sum_{e \in \delta(v)} x_e=0\). We let \(n:=\Abs{\vertexSet}\), and \(\symdif\) denote the symmetric difference operator.

    We denote a (\(uv\)-)walk \(W\) by listing its edges and endpoints sequentially, i.e., by \(W=(u; e_1,\ldots,e_k;v)\). We define its inverse as \(W^{-1}=(v;e_k,\ldots,e_1;u)\). We say a walk is closed if \(u=v\). A trail is a walk in which edges do not repeat. A path is a trail in which internal vertices do not repeat. A cycle is a path which starts and ends at the same vertex. If we refer to the edge set of a walk \(W\), we just write \(W\). Note that this can be a multi-set.

\paragraph{Duality and augmenting structures.} 
    The dual of \eqref{eq:fracmatching} is given by
    \begin{equation} \label{eq:fractional_vertex_cover}
        \taufcG := \min \CBra{ \capacity^\top y + \1^\top z : y_u+y_v+z_{uv}\geq\weight_{uv} \forall\shortEdge{u}{v}\in\edgeSet, y\geq0, z\geq0 },
    \end{equation}
    where \(\1\) is the all-one vector of appropriate size. 
    A solution \((y,z)\) feasible for (\ref{eq:fractional_vertex_cover}) is called a \emph{fractional vertex cover}. By LP theory, we have \(\nucG\leq\nufcG=\taufcG\).
    
   \begin{definition}
    We say that a walk \(W\) is \(\matching\)-alternating (w.r.t. a matching \(\matching\)) if it alternates edges in \(\matching\) and edges not in \(\matching\). We say \(W\) is \(\matching\)-augmenting if it is \(\matching\)-alternating and \(\weight(W\setminus\matching)>\weight(W\cap\matching)\). An \(\matching\)-alternating \(u v\)-walk \(W\) is \emph{proper} if  \(d_u^{W\symdif\matching}\leq c_u\) and \(d_v^{W\symdif\matching}\leq c_v\). 
   \end{definition}
   
    \begin{definition}
        Given an \(\matching\)-alternating walk \(W=(u;e_1,\ldots,\allowbreak e_k;v)\) and an \(\varepsilon>0\), the \emph{\(\varepsilon\)-augmentation} of \(W\) is the vector \(x^{\matching/W}(\varepsilon)\in\R^\edgeSet\) given by
        \begin{equation}
            x_e^{\matching/W}(\varepsilon) = \begin{cases} 
                1-\kappa(e)\varepsilon & \text{ if } e\in\matching, \\
                \kappa(e)\varepsilon & \text{ if } e\notin\matching,
            \end{cases}
        \end{equation}
        where \(\kappa(e)=\Abs{\CBra{i\in[k] \mid e_i=e, e_i\in W}}\). We say that \(W\) is \emph{feasible} if there exists an \(\varepsilon>0\) such that the corresponding \(\varepsilon\)-augmentation of \(W\) is a fractional \(\capacity\)-matching.
    \end{definition}
    
    \begin{remark}
        \label{rmk: feasible distinct endpoints is proper}
        A feasible \(\matching\)-alternating walk with distinct endpoints is proper.
    \end{remark}
    
    \begin{definition}
        An odd cycle \(C=(v;e_1,\ldots,e_k;v)\) is called an \emph{\(\matching\)-blossom} if it is \(\matching\)-alternating such that either \(e_1\) and \(e_k\) are both in \(\matching\), or are both not in \(\matching\). 
        The vertex \(v\) is called the \emph{base} of the blossom. 
    \end{definition}
    \begin{definition}
        An \emph{\(\matching\)-flower} \(C\cup P\) consists of an \(\matching\)-blossom \(C\) with base \(u\) and an \(\matching\)-alternating path \(P=(u;e_1,\ldots,e_k;v)\) such that \((P,C,P^{-1})\) is \(\matching\)-alternating and feasible. The vertex \(v\) is called the \emph{root} of the flower. The flower is \emph{\(\matching\)-augmenting} if
        \begin{equation}
            w(C\setminus\matching) + 2w(P\setminus\matching) > w(C\cap\matching) + 2w(P\cap\matching).
        \end{equation}
    \end{definition}
    \begin{definition}
        An \emph{\(\matching\)-bi-cycle} \(C\cup P\cup D\) consists of two \(\matching\)-blossoms \(C\) and \(D\) with bases \(u\) and \(v\), respectively, and an \(\matching\)-alternating path \(P=(u;e_1,\ldots,e_k;v)\) such that \((P,D,P^{-1},C)\) is \(\matching\)-alternating. The bi-cycle is \emph{\(\matching\)-augmenting} if
        \begin{equation}
            w(C\setminus\matching) + 2w(P\setminus\matching) + w(D\setminus\matching) > w(C\cap\matching) + 2w(P\cap\matching) + w(D\cap\matching).
        \end{equation}
    \end{definition}
    
    Note that, in the last two definitions, it may happen that \(P\) has no edges. In the unit-capacity case it is well-known that a matching \(\matching\) is maximum-weight if and only if there do not exist any proper \(\matching\)-augmenting paths or cycles. This generalizes to the capacitated case. We report a proof for completeness. 
    \begin{theorem}
        \label{thm: M max iff no augmenting trail}
        A \(\capacity\)-matching \(\matching\) in \(\graphwc\) is maximum-weight if and only if \(\graph\) does not contain a proper \(\matching\)-augmenting trail.
    \end{theorem}
    \begin{proof}
        \((\Rightarrow)\)
        If \(\graph\) contains a proper \(\matching\)-augmenting trail \(T\), then \(\matching\symdif T\) is a \(\capacity\)-matching and \(\weight(\matching\symdif T)>\weight(\matching)\), which means \(\matching\) is not maximum-weight.
        
        \((\Leftarrow)\)
        Let \(\matching\) be a \(\capacity\)-matching in \(\graph\) such that \(\matching\) is not maximum-weight. We will show that \(\graph\) contains a proper \(\matching\)-augmenting trail. Let \(N\) be a maximum-weight \(\capacity\)-matching, and consider the graph induced by \(\matching\symdif N\). We construct a unit-capacity graph \(\hat{\graph}\):
        \begin{enumerate}[1.]
            \item 
                For each \(v\in\vertexSet\), define \(b_v:=\max\CBra{d_v^{\matching\setminus N},d_v^{N\setminus\matching}}\), create copies \(v_1,\ldots,v_{b_v}\) and add them to \(\vertexSet(\hat{\graph})\). Initialize \(J_\matching(v)=J_N(v)=\CBra{1,\ldots,b_v}\).
            \item
                For each \(\shortEdge{u}{v}\in\matching\setminus N\), add a single edge \(\shortEdge{u_i}{v_j}\) to both \(\edgeSet(\hat{\graph})\) and \(\hat{\matching}\) with edge-weight \(w_{uv}\), where \(i\in J_\matching(u)\) and \(j\in J_\matching(v)\) are chosen arbitrarily. Remove \(i\) and \(j\) from \(J_\matching(u)\) and \(J_\matching(v)\), respectively.
            \item
                Likewise for each \(\shortEdge{u}{v}\in N\setminus\matching\).
        \end{enumerate}
        Observe that this construction establishes a natural weight-preserving bijection between \(\edgeSet\) and \(\edgeSet(\hat{\graph})\). Furthermore, the sets \(\hat{\matching}\) and \(\hat{N}\) are matchings in \(\hat{\graph}\), and \(\weight(N)>\weight(\matching)\) implies \(\weight(\hat{N})>\weight(\hat{\matching})\). In particular, \(\hat{\matching}\) is not maximum-weight in \(\hat{\graph}\). Since \(\hat{\graph}\) has unit-capacities, it contains at least one proper \(\hat{\matching}\)-augmenting path or cycle \(\hat{T}\). Let \(T=(u;e_1,\ldots,e_k;v)\) be the corresponding \(\matching\)-alternating walk in \(\graph\). Since \(\hat{T}\) does not repeat edges and is actually alternating between \(\hat{\matching}\) and \(\hat{N}\), \(T\) is alternating between \(\matching\setminus N\) and \(N\setminus\matching\), and also does not repeat edges, i.e., \(T\) is a trail. Since \(\weight(\hat{T}\setminus\hat{\matching})>\weight(\hat{T}\cap\hat{\matching})\), we also have \(\weight(T\setminus\matching)>\weight(T\cap\matching)\), i.e., \(T\) is an \(\matching\)-augmenting trail. Thus, we only need to show that \(T\) is proper, i.e., that \(d_u^{T\symdif\matching}\leq c_u\) and \(d_v^{T\symdif\matching}\leq c_v\).
        
        \emph{Case 1: \(\hat{T}\) is a proper \(\hat{\matching}\)-augmenting path.} 
        If \(u=v\), then provided that at least one of \(e_1\) and \(e_k\) is in \(\matching\), then \(d_u^{T\symdif\matching}\leq d_u^\matching\leq c_u\). If on the other hand neither of \(e_1\) and \(e_k\) is in \(\matching\), then the corresponding edges \(\hat{e}_1\) and \(\hat{e}_k\) in \(\hat{T}\) are not in \(\hat{\matching}\), and must therefore be in \(\hat{N}\). Let \(u_i\) and \(u_j\) be the first and last vertices of \(\hat{T}\), incident with \(\hat{e}_1\) and \(\hat{e}_k\), respectively. Note that \(u_i\) and \(u_j\) are distinct, since \(\hat{T}\) is a path. Furthermore, since \(\hat{T}\) is proper, \(u_i\) and \(u_j\) are not incident with edges from \(\hat{\matching}\). Observe that by construction of \(\hat{\graph}\), either all vertices in \(\CBra{u_1,\ldots,u_{b_u}}\) are \(\hat{\matching}\)-covered or \(\hat{N}\)-covered. Since \(u_i\) and \(u_j\) are not \(\hat{\matching}\)-covered, all copies of \(u\) must be \(\hat{N}\)-covered. Hence, \(d_u^{\matching\setminus N}\leq d_u^{N\setminus\matching}-2\), which means \(d_u^{\matching}\leq d_u^{N}-2\). Finally, \(d_u^{T\symdif\matching}=d_u^\matching+2\leq d_u^N\leq c_u\).
        
        If \(u\neq v\), we again consider whether or not \(e_1\) and \(e_k\) are in \(\matching\). Since \(u\) and \(v\) are distinct, these cases for \(e_1\) and \(e_k\) are independent. If \(e_1\in\matching\), then \(d_u^{T\symdif\matching}=d_u^\matching-1\leq c_u\). If \(e_1\notin\matching\), then \(d_u^{T\symdif\matching}=d_u^\matching+1\), and \(\hat{e}_1\) is in \(\hat{N}\), not in \(\hat{\matching}\). Let \(u_i\) be the copy of \(u\) that is incident with \(\hat{e}_1\). Since \(u_i\) is the first vertex of \(\hat{T}\), and \(\hat{T}\) is proper, \(u_i\) is not incident with any edge from \(\hat{\matching}\). By construction of \(\hat{\graph}\), every copy of \(u\) must be \(\hat{N}\)-covered. Hence, \(d_u^{\matching\setminus N}\leq d_u^{N\setminus\matching}-1\), which means \(d_u^{\matching}\leq d_u^{N}-1\). Therefore, \(d_u^{T\symdif\matching}=d_u^\matching+1\leq d_u^{N}\leq c_u\). By symmetry of \(u\) and \(v\), we also have \(d_v^{T\symdif\matching}\leq c_v\) both if \(e_k\in\matching\) and \(e_k\notin\matching\).
        
        \emph{Case 2: \(\hat{T}\) is a \(\hat{\matching}\)-augmenting cycle.} In this case \(u=v\) and exactly one of \(e_1\) and \(e_k\) is in \(\matching\) and one is not, which means \(d_u^{T\symdif\matching}=d_u^\matching\leq c_u\).
    \end{proof}

\paragraph{Auxiliary Construction.}\label{subsec: auxiliary construction}
    We will use a construction given in \cite{Farczadi2013Network}, to transform a pair \(\gwcm\) into another one \(\gwcmaux\), where \(\graph'\) is an auxiliary unit-capacity graph.

    Construction: \(\gwcm\rightarrow\gwcmaux\)
    \begin{enumerate}[1.]
        \item 
            For each \(v\in\vertexSet\), create the set \(C_v=\CBra{v_1,\ldots,v_{\capacity_v}}\) of \(\capacity_v\) copies of \(v\), add \(C_v\) to \(\vertexSet(\graph')\), and initialize \(J(v)=\CBra{1,\ldots,\capacity_v}\).
        \item
            For each \(\shortEdge{u}{v}\in\matching\), add a single edge \(\shortEdge{u_i}{v_j}\) to both \(\edgeSet(\graph')\) and \(\matching'\) with edge-weight \(w_{uv}\), where \(i\in J(u)\) and \(j\in J(v)\) are chosen arbitrarily. Remove \(i\) and \(j\) from \(J(u)\) and \(J(v)\), respectively.
        \item
            For each edge \(\shortEdge{u}{v}\in\edgeSet\setminus\matching\), add an edge \(\shortEdge{u_i}{v_j}\) to \(\edgeSet(\graph')\) with edge-weight \(w_{uv}\), for all \(u_i\in C_u\) and \(v_j\in C_v\).
    \end{enumerate}
    See \cref{fig: aux construction} for an example. 
    In this figure it is easy to see that the matching \(\matching'\) in \(\graph'\) is not maximum, even though \(\matching\) is maximum in \(\graph\).\footnote{It was stated in \cite[corollary 1]{Farczadi2013Network} that \(\matching\) is maximum if and only if \(\matching'\) is maximum, but this example shows this to be false.} 
    
    We define a map \(\eta\) to go back from the auxiliary graph \(\graph'\) to the original graph \(\graph\). Specifically, if \(u_i\in \vertexSet(\graph')\cap C_u\) for some \(u\in\vertexSet\), then \(\eta(u_i):=u\), and if \(\shortEdge{u_i}{v_j}\in\edgeSet(\graph')\) such that \(u_i\in C_u\), \(v_j\in C_v\) for some \(u,v\in\vertexSet\), then \(\eta(\shortEdge{u_i}{v_j}):=\shortEdge{u}{v}\). This extends in the obvious way to paths, cycles, walks, and so on.
    
    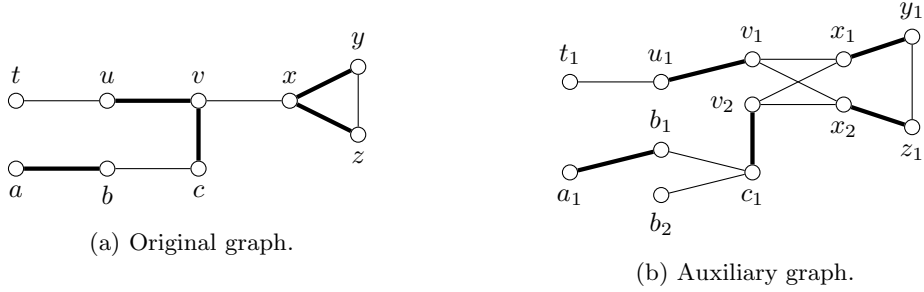
\begin{figure}[t]
        \centering
        \begin{subfigure}{.45\textwidth}
            \centering
            \begin{tikzpicture}[
    scale=.6,
    circ/.style={
        circle,
        draw=black,
        inner sep=2pt
    },
    notMatched/.style={},
    matched/.style={ultra thick},
    fracMatched/.style={dashed,ultra thick}
]

\node[circ,label={above:\(t\)}] (u) at (0,0) {};
\node[circ,label={above:\(u\)}] (v) at (2,0) {};
\node[circ,label={above:\(v\)}] (w) at (4,0) {};
\node[circ,label={above:\(x\)}] (x) at (6,0) {};
\node[circ,label={above:\(y\)}] (y) at (7.5,.75) {};
\node[circ,label={below:\(z\)}] (z) at (7.5,-.75) {};
\node[circ,label={below:\(c\)}] (c) at (4,-1.5) {};
\node[circ,label={below:\(b\)}] (b) at (2,-1.5) {};
\node[circ,label={below:\(a\)}] (a) at (0,-1.5) {};

\draw[notMatched] (u) to (v);
\draw[matched]    (v) to (w);
\draw[notMatched] (w) to (x);
\draw[matched]    (x) to (y);
\draw[matched]    (x) to (z);
\draw[notMatched] (y) to (z);
\draw[matched]    (a) to (b);
\draw[notMatched] (b) to (c);
\draw[matched]    (c) to (w);

\end{tikzpicture}
            \caption{Original graph.}
            \label{fig: aux construction og}
        \end{subfigure}
        \begin{subfigure}{.45\textwidth}
            \centering
            \begin{tikzpicture}[
    scale=.6,
    circ/.style={
        circle,
        draw=black,
        inner sep=2pt
    },
    notMatched/.style={},
    matched/.style={ultra thick},
    fracMatched/.style={dashed,ultra thick}
]

\node[circ,label={above:\(t_1\)}] (u) at (0,0) {};
\node[circ,label={above:\(u_1\)}] (v) at (2,0) {};
\node[circ,label={above:\(v_1\)}] (w1) at (4,.5) {};
\node[circ,label={left:\(v_2\)}] (w2) at (4,-.5) {};
\node[circ,label={above:\(x_1\)}] (x1) at (6,.5) {};
\node[circ,label={below:\(x_2\)}] (x2) at (6,-.5) {};
\node[circ,label={above:\(y_1\)}] (y) at (7.5,1) {};
\node[circ,label={below:\(z_1\)}] (z) at (7.5,-1) {};
\node[circ,label={below:\(c_1\)}] (c) at (4,-2) {};
\node[circ,label={above:\(b_1\)}] (b1) at (2,-1.5) {};
\node[circ,label={below:\(b_2\)}] (b2) at (2,-2.5) {};
\node[circ,label={below:\(a_1\)}] (a) at (0,-2) {};

\draw[notMatched] (u)  to (v);
\draw[matched]    (v)  to (w1);
\draw[notMatched] (w1) to (x1);
\draw[notMatched] (w1) to (x2);
\draw[notMatched] (w2) to (x1);
\draw[notMatched] (w2) to (x2);
\draw[matched]    (x1) to (y);
\draw[matched]    (x2) to (z);
\draw[notMatched] (y)  to (z);
\draw[matched]    (w2) to (c);
\draw[notMatched] (b1) to (c);
\draw[notMatched] (b2) to (c);
\draw[matched]    (b1) to (a);

\end{tikzpicture}
            \caption{Auxiliary graph.}
            \label{fig: aux construction aux}
        \end{subfigure}
        \caption{Example of the auxiliary construction on an instance \(\gwcm\). Capacities are all 1 except for \(\capacity_v=\capacity_x=\capacity_b=2\). Weights are all 1 except for \(\weight_{bc}=0.5\). The matching is displayed as bold edges.}
        \label{fig: aux construction}
    \end{figure}
    
    \begin{remark}\label{rmk: removing vertices in og and aux}
        If \(\gwcm\) has auxiliary \(\gwcmaux\), and \(X\subseteq\vertexSet\) is any set of vertices which avoids \(\matching\), then \((\graph\setminus X)'=\graph'\setminus X'\), where \(X'=\cup_{v\in X} C_v\).
    \end{remark}
    
    The following easy lemma will be useful.
    \begin{lemma}
        \label{lem: mapping}
        Given \(\gwcm\) and auxiliary \(\gwcmaux\), let \(P\) be a feasible \(\matching'\)-augmenting walk. Then, \(\eta(P)\) is a feasible \(\matching\)-augmenting walk.
    \end{lemma}
    \begin{proof}
        Let \(e_1=\shortEdge{u}{v}\) and \(e_2=\shortEdge{v}{w}\) be two consecutive edges on \(P\). Then \(\eta(e_1)\) and \(\eta(e_2)\) are the corresponding edges on \(\eta(P)\), and they are both incident with \(\eta(v)\). Hence, \(\eta(P)\) is a walk.
        
        For any edge \(e\) on \(P\), we have \(e\in\matching'\) if and only if \(\eta(e)\in\matching\). In addition, \(\weight'_e=\weight_{\eta(e)}\). So, \(\eta(P)\) is an \(\matching\)-augmenting walk.
        
        Suppose \(P=(u;e_1,\ldots,e_k;v)\). Feasibility of \(P\) means that either \(e_1\in\matching'\), or \(u\) is \(\matching'\)-exposed. Likewise for \(e_k\) and \(v\). It follows that either \(\eta(e_1)\in\matching\), or \(\eta(u)\) is \(\matching\)-unsaturated. Likewise for \(\eta(e_k)\) and \(\eta(v)\). This means \(\eta(P)\) is feasible.
    \end{proof}
    
    We will need the following theorem.
    
    \begin{theorem}\label{thm: G not stable iff G' contains}
       An \(\matching\)-vertex-stabilizer instance \(\gwcm\) is not stable if and only if the graph \(\graph'\) in the auxiliary construction \(\gwcmaux\) contains at least one of the following: (i) an \(\matching'\)-augmenting flower; (ii) an \(\matching'\)-augmenting bi-cycle; (iii) a proper \(\matching'\)-augmenting path; (iv) an \(\matching'\)-augmenting cycle.
    \end{theorem}
    \begin{proof} 
        It was proven in \cite[theorem 2]{Farczadi2013Network} that \(\gwcm\) does not correspond to a stable \(\matching\)-vertex-stabilizer instance if and only if \(\gwcmaux\) does not correspond to a stable \(\matching'\)-vertex-stabilizer instance. We distinguish two scenarios for when the latter condition occurs. If \(\matching'\) is maximum-weight, then \(\graph'\) contains an \(\matching'\)-augmenting flower or bi-cycle, see \cite[theorem 1]{Koh2020Stabilizing}. If \(\matching'\) is not maximum-weight, \(\graph'\) must contain a proper \(\matching'\)-augmenting path or cycle, by standard matching theory.
    \end{proof}
    
    We will refer to an augmenting structure of type \((i)-(iv)\) in \cref{thm: G not stable iff G' contains} as a \emph{basic} augmenting structure. The next lemma follows from \cite{Koh2020Stabilizing}.
    
    \begin{lemma}\label{lem: decomposition}
        Let \(\graph'\) be a unit-capacity graph, and \(\matching'\) be any (not necessarily maximum) matching of \(\graph'\).
        \begin{itemize} 
            \item[(a)] For any \(\matching'\)-exposed vertex \(u\), one can compute a feasible \(\matching'\)-augmenting walk starting at \(u\) of length at most \(3\Abs{\vertexSet(\graph')}\), or determine that none exists, in polynomial time.
            \item[(b)] A feasible \(\matching'\)-augmenting \(u v\)-walk contains a feasible \(\matching'\)-augmenting \(u v\)-path (proper if \(u\neq v\)), an \(\matching'\)-augmenting cycle, an \(\matching'\)-augmenting flower rooted at \(u\) or \(v\), or an \(\matching'\)-augmenting bi-cycle. Furthermore, this augmenting structure can be computed in polynomial time.
        \end{itemize}
    \end{lemma}
    \begin{proof}
        (a) When given a graph \(\graph'\), a matching \(\matching'\), a vertex \(u\), and an integer \(k\), algorithm 3 in \cite{Koh2020Stabilizing} computes a feasible \(\matching'\)-augmenting \(uv\)-walk of length at most \(k\), or determines none exist, for all \(v\in\vertexSet(\graph')\). Lemma 7 and 8 in \cite{Koh2020Stabilizing} show correctness of the algorithm. The algorithm is polynomial time in \(k\), \(\Abs{\vertexSet(\graph')}\), and \(\Abs{\edgeSet(\graph')}\). Since we use \(k=3\Abs{\vertexSet(\graph')}\), it is polynomial time.  
        As mentioned, algorithm 3 in \cite{Koh2020Stabilizing} actually returns one \(u v\)-walk per \(v\in\vertexSet(\graph')\), if at least one exists for \(v\). We just need one such walk, so if for at least one \(v\) a \(u v\)-walk is returned, we arbitrarily choose one, otherwise we know no such walk starting at \(u\) exists.
        
        (b)
        Lemma 9 in \cite{Koh2020Stabilizing} directly gives us that a feasible \(\matching'\)-augmenting \(u v\)-walk contains a feasible \(\matching'\)-augmenting \(u v\)-path, an \(\matching'\)-augmenting cycle, an \(\matching'\)-augmenting flower rooted at \(u\) or \(v\), or an \(\matching'\)-augmenting bi-cycle. By \cref{rmk: feasible distinct endpoints is proper} the path is proper if \(u\neq v\). Lemma 9 in \cite{Koh2020Stabilizing} is proven in a constructive way, hence it also gives a way to compute the augmenting structure in polynomial time.
    \end{proof}
    
    The next theorem is standard.
    \begin{theorem}
        \label{thm: feasible augmenting walk not stable}
        An \(\matching\)-vertex-stabilizer instance \(\gwcm\) is stable if and only if \(\graph\) does not contain a feasible \(\matching\)-augmenting walk.
    \end{theorem}
    \begin{proof}
        We prove both directions by contraposition.
        
        (\(\Rightarrow\))
        Assume there exists a feasible \(\matching\)-augmenting walk \(W\). Since \(W\) is augmenting, \(\weight(W\setminus\matching)>\weight(W\cap\matching)\), and since \(W\) is feasible, \(x^{\matching/W}(\varepsilon)\) is a fractional \(\capacity\)-matching. Together they imply
        \begin{equation}
            \nufcG \geq \weight^\top x^{\matching/W}(\varepsilon) 
                = \weight(\matching) - \varepsilon \weight(W\cap\matching) + \varepsilon \weight(W\setminus\matching)
                > \weight(\matching),
        \end{equation}
        i.e., the instance \(\gwcm\) is not stable.

        (\(\Leftarrow\))
        Assume the instance is not stable. Then by \cref{thm: G not stable iff G' contains}, the graph \(\graph'\) from the auxiliary \(\gwcmaux\) contains a basic augmenting structure, which clearly is a feasible \(\matching'\)-augmenting walk \(P\). Then \(\eta(P)\) is a feasible \(\matching\)-augmenting walk, by \cref{lem: mapping}.
    \end{proof}
\section{\texorpdfstring{\(M\)}{M}-vertex-stabilizer}\label{sec: M-vertex-stab}

The goal of this section is to prove the following theorem.

\begin{theorem}
    \label{thm: M-vertex-stab poly-time}
    The \(\matching\)-vertex-stabilizer problem on weighted, capacitated graphs can be solved in polynomial time.
\end{theorem}

\paragraph{Overview of the strategy.} 
    A natural strategy would be to first apply the auxiliary construction described in \cref{subsec: auxiliary construction} to reduce to unit-capacity instances, and then apply the algorithm proposed in \cite{Koh2020Stabilizing} which solves the problem exactly. However, there is a critical issue with this strategy. Namely, the auxiliary construction applied to unstable instances does \emph{not} always preserve maximality of the corresponding matchings, as shown in \cref{fig: aux construction}. In that example, the matching \(\matching'\) is not maximum in \(\graph'\). The algorithm of \cite{Koh2020Stabilizing}, if applied to an instance where the given matching is not maximum, is not guaranteed to find an optimal solution, but only a 2-approximate one (see theorem 12 in \cite{Koh2020Stabilizing}). In addition, since the auxiliary construction splits a vertex into multiple ones, we may even get infeasible solutions. As a concrete example of this, the algorithm of \cite{Koh2020Stabilizing} applied to the instance of \cref{fig: aux construction aux} will include \(b_2\) in its proposed solution. Mapping this solution to our capacitated instance would imply to remove \(b\), which is clearly not allowed as \(b\) is \(\matching\)-covered.
    
    To overtake this issue, we do not apply the algorithm of \cite{Koh2020Stabilizing} as a black-box, but use parts of it (highlighted in \cref{lem: decomposition}) in a careful way. In particular, we use it to compute a sequence of feasible augmenting walks in \(\graph'\). We actually show that the walks in \(\graph'\) which might create the issue described before when mapped backed to \(\graph\), are the walks in which at least one edge of \(\graph\) is traversed more than once in opposite directions, and that have two distinct endpoints. When this happens, we prove that we can modify the walk and get one where the endpoints coincide, which will still be feasible and augmenting. In this latter case, we can then either correctly identify a vertex to remove (the unique endpoint), or determine that the instance cannot be stabilized.

\paragraph{A more detailed description.} 
    We start by defining \emph{ties}.
    \begin{definition}
        Given \(\gwcm\) with auxiliary \(\gwcmaux\), and an \(\matching'\)-alternating path \(P'\), a \emph{tie} in \(P'\) is a pair of unmatched edges \(\CBra{\shortEdge{a}{b},\shortEdge{c}{d}}\) on \(P'\) such that for some distinct \(u,v\in\vertexSet\), either \(\CBra{a,c}\subseteq C_u\) and \(\CBra{b,d}\subseteq C_v\) or \(\CBra{a,d}\subseteq C_u\) and \(\CBra{b,c}\subseteq C_v\). We say \(P'\) is \emph{tieless} if it does not contain a tie.
    \end{definition}

    We now show that if the auxiliary construction does not preserve maximality of the \(\capacity\)-matching \(\matching\), then we must have ties in all proper \(\matching'\)-augmenting paths and cycles. 

    \begin{lemma}
        \label{thm: M max M' not}
        Given \(\gwcm\) with auxiliary \(\gwcmaux\), if \(\matching\) is a maximum-weight \(\capacity\)-matching in \(\graph\), then all proper \(\matching'\)-augmenting paths and cycles contain ties.
    \end{lemma}
    \begin{proof}
        We prove this by contraposition. So, suppose that there is a proper \(\matching'\)-augmenting path or cycle \(P'\) that is tieless. Note that \(P'\) is also feasible. By \cref{lem: mapping}, \(P=\eta(P')\) is a feasible \(\matching\)-augmenting walk. Since \(P'\) is tieless, there is a bijection between \(\edgeSet(P')\) and \(\edgeSet(P)\), and so, as \(P'\) does not repeat edges, neither does \(P\). Hence \(P\) is a feasible \(\matching\)-augmenting trail. We will show that \(P\) is proper.
        
        If \(P'\) is an \(\matching'\)-augmenting cycle, \(P\) is a closed \(\matching\)-augmenting trail of even length. It follows that \(d_v^{P\symdif\matching}=d_v^{\matching}\leq\capacity_v\) for all vertices \(v\) on \(P\), and hence \(P\) is proper.
        
        Now suppose \(P'\) is a proper \(\matching'\)-augmenting path. Let \(P'=(u_i;e_1',\ldots,e_k';v_j)\) and \(u=\eta(u_i)\), \(v=\eta(v_j)\), \(e_1=\eta(e_1')\) and \(e_k=\eta(e_k')\). Note that, because \(P'\) is proper, \(e_1'\notin\matching\) if and only if \(u_i\) is \(\matching'\)-exposed. Likewise for \(e_k'\) and \(v_j\).
        
        \textit{Case 1: \(u=v\).}
        If at most one of \(u_i\) and \(v_j\) is \(\matching'\)-exposed, then at least one of \(e_1'\) and \(e_k'\) is in \(\matching'\) and hence at least one of \(e_1\) and \(e_k\) is in \(\matching\). Therefore, \(d_u^{P\symdif\matching}\leq d_u^\matching\leq \capacity_u\). If both \(u_i\) and \(v_j\) are \(\matching'\)-exposed, then \(e_1',e_k'\notin\matching'\) and hence \(e_1,e_k\notin\matching\). Therefore, \(d_u^{P\symdif\matching}=d_u^\matching+2\).
        By construction there are \(c_u\) copies of \(u\), and since \(u_i\) and \(v_j\) are already two of those copies, and they are exposed, we have \(d_u^\matching\leq c_u-2\). Thus \(d_u^{P\symdif\matching}\leq c_u\).
        
        \textit{Case 2: \(u\neq v\).}
        If \(e_1'\in\matching'\), then \(e_1\in\matching\), and so we have \(d_u^{P\symdif\matching}=d_u^\matching-1\leq c_u\). If \(e_1'\notin\matching'\), then \(e_1\notin\matching\), and so we have \(d_u^{P\symdif\matching}=d_u^\matching+1\). Using the same reasoning as in case 1, we can conclude that \(d_u^\matching\leq c_u-1\) because \(u_i\) is \(\matching'\)-exposed, and therefore \(d_u^{P\symdif\matching}\leq c_u\). The argument is analogous for \(v\).
        
        In all cases \(P\) is a proper \(\matching\)-augmenting trail. It follows by \cref{thm: M max iff no augmenting trail} that \(\matching\) is not a maximum-weight \(\capacity\)-matching in \(\graph\).
    \end{proof}    

    We now define the operation of \emph{traceback}, which we will use to modify the feasible augmenting walks, when needed.

    \begin{definition}
        Given \(\gwcm\) and an \(\matching\)-alternating walk \(P=(u;e_1,\ldots,\allowbreak e_k;\allowbreak v)\) which repeats an edge in opposite directions, let \(t\) be the least index such that \(e_t=e_s\) for some \(s<t\), and \(e_s\) and \(e_t\) are traversed in opposite directions by \(P\). Then the \emph{\(u\)-traceback} and \emph{\(v\)-traceback} of \(P\) are defined as the walks \(\tb(P,u)=(e_1,\ldots,e_t,e_{s-1},e_{s-2},\ldots,e_1)\) and \(\tb(P,v)=(e_k,e_{k-1}\ldots,e_s,e_{t+1},e_{t+2},\ldots,e_k)\).
    \end{definition}
    
    The next lemma explains how to use the traceback operation. 
    \begin{lemma}\label{lem: traceback}
        Given \(\gwcm\) such that \(\matching\) is maximum-weight, and auxiliary \(\gwcmaux\), let \(P'=(u_i;e_1',\allowbreak\ldots,e_k';v_j)\) be a proper \(\matching'\)-augmenting path such that both \(u_i\) and \(v_j\) are \(\matching'\)-exposed and \(\eta(u_i)\neq\eta(v_j)\). Then \(\tb(\eta(P'),\eta(u_i))\) and \(\tb(\eta(P'),\eta(v_j))\) are well-defined, feasible \(\matching\)-alternating walks, and at least one of them is \(\matching\)-augmenting.
    \end{lemma}
    \begin{proof}
        Let \(P=\eta(P')=(u;e_1,\ldots,e_k,v)\). By \cref{lem: mapping}, \(P\) is a feasible \(\matching\)-augmenting walk, and even proper by \cref{rmk: feasible distinct endpoints is proper}, since \(u\neq v\). To show that \(\tb(P,u)\) and \(\tb(P,v)\) are well-defined, we must show that \(P\) traverses some edge in opposite directions. By \cref{thm: M max M' not} we already have that \(P'\) contains a tie, and hence that \(P\) traverses some edge twice. 
        We now show that there must exist at least one edge that is traversed in opposite direction. Suppose not, let \(t\) be the least index such that \(e_t=e_s\) for some \(s<t\). Decompose \(P\) as \((P_1,e_s,P_2,e_t,P_3)\).
        \begin{claim}
            If \(P\) traverses \(e_s\) and \(e_t\) in the same direction, then \((P_1,e_s,P_3)\) is a shorter proper \(\matching\)-augmenting walk.
        \end{claim}
        \begin{proof}
            For notation, define \(P_2^+=(P_2,e_t)\). By definition, \(P_2^+\) is an \(\matching\)-alternating closed trail of even length. It follows that \(d_v^{P_2^+\symdif\matching}=d_v^{\matching}\leq\capacity_v\) for all vertices \(v\) on \(P_2^+\), and hence \(P_2^+\) is proper. Since \(\matching\) is maximum-weight, \cref{thm: M max iff no augmenting trail} implies that \(P_2^+\) cannot be \(\matching\)-augmenting. However, \(P\) is \(\matching\)-augmenting, which means the augmenting part must come from \(P\setminus P_2^+\). Hence, \((P_1,e_s,P_3)\) is an \(\matching\)-augmenting walk. It is proper because \(P\) is proper.
        \end{proof}
        By this claim \(W=(P_1,e_s,P_3)\) is a shorter proper \(\matching\)-augmenting walk. \(W\) also necessarily repeats an edge, because otherwise \(W\) is a proper \(\matching\)-augmenting trail, contradicting that \(\matching\) is maximum-weight, by \cref{thm: M max iff no augmenting trail}. Then we can apply the claim again, to find an even shorter proper \(\matching\)-augmenting walk. This argument can be repeated until eventually we reach a contradiction.
        
        Thus there is at least one edge traversed in opposite direction, hence \(\tb(P,u)\) and \(\tb(P,v)\) are well-defined. Clearly \(\tb(P,u)\) and \(\tb(P,v)\) are \(\matching\)-alternating. Furthermore, since \(u_i\) and \(v_j\) are \(\matching'\)-exposed, \(u\) and \(v\) are \(\matching\)-unsaturated. It follows that \(\tb(P,u)\) and \(\tb(P,v)\) are feasible.

        That leaves to show that at least one of them is \(\matching\)-augmenting. For notation, let \(t\) be the least index such that \(e_t=e_s\) for some \(s<t\) and \(e_t\) and \(e_s\) are traversed in opposite direction. As before, decompose \(P\) as \((P_1,e_s,P_2,e_t,P_3)\). Define \(P_2^{++}=(e_s,P_2,e_t)\), \(P_u=\tb(P,u)\), and \(P_v=\tb(P,v)\). Note that \(P_u=(P_1,P_2^{++},P_1^{-1})\) and \(P_v=(P_3^{-1},(P_2^{++})^{-1},P_3)\).
        
        \textit{Case 1: \(\weight(P_1\setminus\matching)-\weight(P_3\setminus\matching)>\weight(P_1\cap\matching)-\weight(P_3\cap\matching)\).}
        Because \(P\) is \(\matching\)-augmenting, we know that
        \begin{equation}
            \weight(P_1\setminus\matching)+\weight(P_2^{++}\setminus\matching)+\weight(P_3\setminus\matching) > \weight(P_1\cap\matching)+\weight(P_2^{++}\cap\matching)+\weight(P_3\cap\matching).
        \end{equation}
        Adding these inequalities, we obtain
        \begin{equation}
            \weight(P_u\setminus\matching)=2\weight(P_1\setminus\matching)+\weight(P_2^{++}\setminus\matching) > 2\weight(P_1\cap\matching)+\weight(P_2^{++}\cap\matching)=\weight(P_u\cap\matching).
        \end{equation}
        Hence, \(P_u\) is \(\matching\)-augmenting.
        
        \textit{Case 2:  \(\weight(P_1\setminus\matching)-\weight(P_3\setminus\matching)<\weight(P_1\cap\matching)-\weight(P_3\cap\matching)\).}
        Analogous to case 1, we find that \(P_v\) is \(\matching\)-augmenting.
        
        \textit{Case 3:  \(\weight(P_1\setminus\matching)-\weight(P_3\setminus\matching)=\weight(P_1\cap\matching)-\weight(P_3\cap\matching)\).}
        Analogous to case 1, we find that both \(P_u\) and \(P_v\) are \(\matching\)-augmenting.
    \end{proof}

\begin{algorithm}[htp]
    \caption{finding an \(\matching\)-vertex-stabilizer}
    \label{alg: M-vertex-stab}
    
    \SetKwInOut{Input}{input}
    \Input{\([\graphwc,\matching]\)}

    compute the auxiliary \(\gwcmaux\) \\
    initialize \(S\leftarrow\emptyset\), \(L\leftarrow M'\)-exposed vertices \\
    \While{ \(L \neq \emptyset\)}{
        select \(u_i \in L\) and compute a feasible \(\matching'\)-augmenting walk starting at \(u_i\) using \cref{lem: decomposition}(a) \\
        \If{no such walk exists}{
            \(L \leftarrow L \setminus \{u_i\}\)
        }
        \Else{
            consider the computed feasible \(\matching'\)-augmenting \(u_i v_j\)-walk\\
            \uIf{both \(\eta(u_i)\) and \(\eta(v_j)\) are \(\matching\)-covered}{
                \Return infeasible
            }
            \uElseIf{\(\eta(u_i)\) is \(\matching\)-covered and \(\eta(v_j)\) is not}{
                \(S\leftarrow S\cup\eta(v_j)\), \(\graph\leftarrow\graph\setminus\eta(v_j)\), \(\graph'\leftarrow\graph'\setminus C_{\eta(v_j)}\), \(L\leftarrow L\setminus C_{\eta(v_j)}\)
            }
            \uElseIf{\(\eta(v_j)\) is \(\matching\)-covered and \(\eta(u_i)\) is not}{
                \(S\leftarrow S\cup\eta(u_i)\),  \(\graph\leftarrow\graph\setminus\eta(u_i)\), \(\graph'\leftarrow\graph'\setminus C_{\eta(u_i)}\), \(L\leftarrow L\setminus C_{\eta(u_i)}\)
            }
            \Else{
                \uIf{\(\eta(u_i)=\eta(v_j)\)}{
                    \(S\leftarrow S\cup\eta(u_i)\), \(\graph\leftarrow\graph\setminus\eta(u_i)\), \(\graph'\leftarrow\graph'\setminus C_{\eta(u_i)}\), \(L\leftarrow L\setminus C_{\eta(u_i)}\)
                }
                \Else{
                    find a basic \(\matching'\)-augmenting structure \(W\) contained in the \(u_i v_j\)-walk using \cref{lem: decomposition}(b) \\
                    \If{\(W\) is an \(\matching'\)-augmenting cycle or bi-cycle}{
                        \Return infeasible
                    }
                    \If{\(W\) is an \(\matching'\)-augmenting flower rooted at \(u_i\)}{
                        \(S\leftarrow S\cup\eta(u_i)\), \(\graph\leftarrow\graph\setminus\eta(u_i)\), \(\graph'\leftarrow\graph'\setminus C_{\eta(u_i)}\), \(L\leftarrow L\setminus C_{\eta(u_i)}\)
                    }
                    \If{\(W\) is an \(\matching'\)-augmenting flower rooted at \(v_j\)}{
                        \(S\leftarrow S\cup\eta(v_j)\), \(\graph\leftarrow\graph\setminus\eta(v_j)\), \(\graph'\leftarrow\graph'\setminus C_{\eta(v_j)}\), \(L\leftarrow L\setminus C_{\eta(v_j)}\)
                    }
                    \If{\(W\) is a proper \(\matching'\)-augmenting \(u_i v_j\)-path}{
                        compute \(\tb(\eta(W),\eta(u_i))\) and \(\tb(\eta(W),\eta(v_j))\) \\
                        \If{\(\tb(\eta(W),\eta(u_i))\) is \(\matching\)-augmenting}{
                            \(S\leftarrow S\cup\eta(u_i)\), \(\graph\leftarrow\graph\setminus\eta(u_i)\), \(\graph'\leftarrow\graph'\setminus C_{\eta(u_i)}\), \(L\leftarrow L\setminus C_{\eta(u_i)}\)
                        }
                        \If{\(\tb(\eta(W),\eta(v_j))\) is \(\matching\)-augmenting}{
                            \(S\leftarrow S\cup\eta(v_j)\), \(\graph\leftarrow\graph\setminus\eta(v_j)\), \(\graph'\leftarrow\graph'\setminus C_{\eta(v_j)}\), \(L\leftarrow L\setminus C_{\eta(v_j)}\)
                        }
                    }
                }
            }
        }
    }
    \uIf{\(\weight(\matching)<\nu^c_f(G)\)}{
        \Return infeasible
    }
    \Else{
        \Return \(S\)
    }
\end{algorithm}

\begin{proof}[Proof of \cref{thm: M-vertex-stab poly-time}]
    Let \(\gwcm\) be the input for the \(\matching\)-vertex-stabilizer problem, with auxiliary \(\gwcmaux\). \Cref{alg: M-vertex-stab} iteratively considers an \(\matching'\)-exposed vertex \(u_i\), and computes a feasible \(\matching'\)-augmenting walk \(U\) starting at \(u_i\), if one exists. \Cref{lem: mapping} implies that \(\eta(U)\) is a feasible \(\matching\)-augmenting walk in \(\graph\). \Cref{thm: feasible augmenting walk not stable} implies that we need to remove at least one vertex of the walk \(\eta(U)\) to stabilize the instance. Note that every vertex \(a \neq u_i,v_j\) of \(U\) is \(\matching'\)-covered, and hence, \(\eta(a)\) is \(\matching\)-covered. Therefore, the only vertices we can potentially remove are \(\eta(u_i)\) or \(\eta(v_j)\). Hence, if both \(\eta(u_i)\) and \(\eta(v_j)\) are \(\matching\)-covered, the instance cannot be stabilized and \cref{alg: M-vertex-stab} checks this in line 9. If only one among \(\eta(u_i)\) and \(\eta(v_j)\) is \(\matching\)-covered, then necessarily we have to remove the \(\matching\)-exposed vertex among the two. \Cref{alg: M-vertex-stab} checks this in line 11 and 13. Note that, by \cref{rmk: removing vertices in og and aux}, instead of computing a new auxiliary for the modified \(\graph\), we can just remove \(C_{\eta(u_i)}\) (resp. \(C_{\eta(v_j)}\)) from \(\graph'\). Similarly, if \(\eta(u_i) = \eta(v_j)\) and \(\eta(u_i)\) is \(\matching\)-exposed, we necessarily have to remove \(\eta(u_i)\). \Cref{alg: M-vertex-stab} checks this in line 16. If instead \(\eta(u_i)\neq\eta(v_j)\), and both are \(\matching'\)-exposed, we apply \cref{lem: decomposition}(b) to find a basic augmenting structure \(W\) contained in \(U\). Once again, we know by \cref{lem: mapping} and \cref{thm: feasible augmenting walk not stable} that we need to remove a vertex in \(\eta(W)\). In case \(W\) is a cycle or bi-cycle, all vertices of \(\eta(W)\) are \(\matching\)-covered so the instance cannot be stabilized and \cref{alg: M-vertex-stab} checks this in line 20. In case \(W\) is a \(\matching'\)-augmenting flower with base \(u_i\) or \(v_j\), \cref{alg: M-vertex-stab} accordingly removes \(\eta(u_i)\) or \(\eta(v_j)\) as all other vertices in \(\eta(W)\) are \(\matching\)-covered, in line 23 and 25. Finally, if \(W\) is a proper (because \(\eta(u_i)\neq\eta(v_j)\)) \(\matching'\)-augmenting path, by \cref{lem: traceback} we know that we can find a feasible \(\matching\)-augmenting walk, where the only \(\matching\)-exposed vertex is either \(\eta(u_i)\) or \(\eta(v_j)\). Once again, this implies that this vertex must be removed. \Cref{alg: M-vertex-stab} does so in lines 29 and 31. 
    
    From the discussion so far, it follows that when we exit the while loop each vertex in \(S\) is a necessary vertex to be removed from \(\graph\), in order to stabilize the instance. We now argue that either removing all vertices in \(S\) is also sufficient, or the instance cannot be stabilized. Suppose that the \(\matching\)-vertex-stabilizer instance given by \(\graph\setminus S\) and \(\matching\) is not stable. \Cref{thm: G not stable iff G' contains} implies that \((\graph\setminus S)'\) contains a basic augmenting structure \(Q\). Note that \(Q\) cannot be an \(\matching'\)-augmenting flower with exposed root, or a proper \(\matching'\)-augmenting path with at least one exposed endpoint. To see this, observe that a flower and path are feasible \(\matching'\)-augmenting walks of length at most \(3\Abs{\vertexSet(\graph')}\) and \(\Abs{\vertexSet(\graph')}\), respectively. Hence, they would have been found by \cref{alg: M-vertex-stab} in line 4, contradicting that \(Q\) exists in \((\graph\setminus S)'\). It follows that \(Q\) is a basic augmenting structure where all vertices are \(\matching'\)-covered. By \cref{lem: mapping} \(\eta(Q)\) is a feasible \(\matching\)-augmenting walk where all vertices are \(\matching\)-covered. This implies that the instance cannot be stabilized. Furthermore, using the \(\varepsilon\)-augmentation of \(\eta(Q)\) we can obtain a fractional \(\capacity\)-matching whose value is strictly greater than \(\weight(\matching)\). Hence, \(\weight(\matching) < \nu^\capacity_f(\graph \setminus S)\). \Cref{alg: M-vertex-stab} correctly determines this in line 32. This proves correctness of our algorithm. 
    
    Finally, we argue about the running time of the algorithm. Note that each operation that the algorithm performs can be done in polynomial time. Furthermore, after each iteration of the while loop, we either determine that the instance cannot be stabilized, or remove a vertex from \(\graph\).  Therefore, the while loop can be executed at most \(n\) times. The result follows.
\end{proof}

We close this section with a remark. The authors in \cite{Koh2020Stabilizing} have also considered the following problem: given a weighted graph \(\graph\) and a (non necessarily maximum-weight) matching \(\matching\), find a minimum-cardinality \(S\subseteq\vertexSet\) such that \(\graph\setminus S\) is stable, and \(\matching\) is a maximum-weight matching in \(\graph\setminus S\), i.e., such that the \(\matching\)-vertex-stabilizer instance given by \(\graph\setminus S\) and \(\matching\) is stable. This is a generalization of our definition of the \(\matching\)-vertex-stabilizer problem, which essentially allows \(\matching\) to be not maximum-weight\footnote{In fact, this is the way the \(\matching\)-vertex-stabilizer problem is defined in~\cite{Koh2020Stabilizing}. We instead use the original definition in \cite{Ahmadian2018Stabilizing,CHANDRASEKARAN201956} which assumes \(\matching\) to be maximum. }. The authors show that this problem is NP-hard, but admits a 2-approximation algorithm (we mentioned this in the strategy overview), which is best possible assuming Unique Game Conjecture. 

With a minor modification of \cref{alg: M-vertex-stab}, we can generalize this result to the capacitated setting. Specifically, we start the algorithm by checking if \(\matching\) is maximum-weight in \(\graph\), and store this in the indicator variable \(\matching_{\max}\). Then, we replace lines 26-31 by \cref{alg: M-vertex-stab modification}.
\begin{algorithm}
    \caption{modification \cref{alg: M-vertex-stab}, lines 26-31}
    \label{alg: M-vertex-stab modification}
    
    \If{\(W\) is a proper \(\matching'\)-augmenting \(u_i v_j\)-path}{
        \If{\(\matching_{\max}\)}{
            compute \(\tb(\eta(W),\eta(u_i))\) and \(\tb(\eta(W),\eta(v_j))\) \\
            \If{\(\tb(\eta(W),\eta(u_i))\) is \(\matching\)-augmenting}{
                \(S\leftarrow S\cup\eta(u_i)\), \(\graph\leftarrow\graph\setminus\eta(u_i)\), \(\graph'\leftarrow\graph'\setminus C_{\eta(u_i)}\), \(L\leftarrow L\setminus C_{\eta(u_i)}\)
            }
            \If{\(\tb(\eta(W),\eta(v_j))\) is \(\matching\)-augmenting}{
                \(S\leftarrow S\cup\eta(v_j)\), \(\graph\leftarrow\graph\setminus\eta(v_j)\), \(\graph'\leftarrow\graph'\setminus C_{\eta(v_j)}\), \(L\leftarrow L\setminus C_{\eta(v_j)}\)
            }
        }
        \Else{
            \(S\leftarrow S\cup\CBra{\eta(u_i),\eta(v_j)}\), \(\graph\leftarrow\graph\setminus\CBra{\eta(u_i),\eta(v_j)}\), \(\graph'\leftarrow\graph'\setminus \Par{C_{\eta(u_i)}\cup C_{\eta(v_j)}}\), \(L\leftarrow L\setminus \Par{C_{\eta(u_i)}\cup C_{\eta(v_j)}}\)
        }
    }
\end{algorithm}

\begin{theorem}
    \label{thm: M-vertex-stab generalized}
    Given a weighted, capacitated graph \(\graphVE\) and a \(c\)-matching \(\matching\), the problem of computing a minimum-cardinality \(S\subseteq\vertexSet\) such that
    \(\graph\setminus S\) is stable, and \(\matching\) is a maximum-weight \(\capacity\)-matching in \(\graph\setminus S\), admits an efficient 2-approximation algorithm.
\end{theorem}
\begin{proof}
    If \(\matching\) is a maximum-weight \(\capacity\)-matching in \(\graph\), the algorithm is unchanged, and the result follows from \cref{thm: M-vertex-stab poly-time}. So suppose \(\matching\) is not maximum-weight. We follow the argument as in the proof of \cref{thm: M-vertex-stab poly-time} (\cref{thm: feasible augmenting walk not stable} still holds if \(\matching\) is not maximum-weight in \(\graph\)), until we reach the case where we have a proper \(\matching'\)-augmenting path \(W\). We know by \cref{lem: mapping} and \cref{thm: feasible augmenting walk not stable} that we need to remove a vertex in \(\eta(W)\). We have \(\eta(u_i)\neq\eta(v_j)\) and both are \(\matching\)-exposed. Even though it might only be necessary to remove one of them, \cref{alg: M-vertex-stab modification} removes both vertices on line 9.
    
    Using the same argumentation as in the proof of \cref{thm: M-vertex-stab poly-time} (\cref{thm: G not stable iff G' contains} still holds if \(\matching\) is not maximum-weight in \(\graph\)), we conclude that either removing all vertices in \(S\) is sufficient, or the instance cannot be stabilized, and the algorithm correctly determines this. But, in contrast to that proof, when we exit the while loop, each vertex in \(S\) is either a necessary vertex to be removed from \(\graph\), in order to stabilize the instance, or it was one of two vertices for which it was necessary to remove at least one. Therefore, for any \(\matching\)-vertex-stabilizer \(S^*\) we have \(\Abs{S^*}\geq\frac{1}{2}\Abs{S}\). It follows that \cref{alg: M-vertex-stab} with the described modification is a 2-approximation.
    
    The modifications are all operations that can be done in polynomial time. The result follows.
\end{proof}
\section{Vertex-Stabilizer}\label{sec: vertex-stab}

The goal of this section is to prove the following theorem.

\begin{theorem}\label{thm: complexity vertex stab G w c}
    The vertex-stabilizer problem on capacitated graphs is \NP-complete, even if all edges have unit-weight. Furthermore, no efficient \(n^{1-\varepsilon}\)-approximation exists for any \(\varepsilon>0\), unless \(\Po=\NP\).
\end{theorem}

Note that, given an unstable graph \(\graphwc\), removing all vertices (but two) trivially yields a stable graph. This gives a (trivial) \(n\)-approximation algorithm for the vertex-stabilizer problem. The theorem above essentially implies that one cannot hope for a much better approximation. To prove it, we will use:

\begin{namedthm}[Minimum Independent Dominating Set (MIDS) problem]
    Given a graph \(\graphVE\), compute a minimum-cardinality subset \(S \subseteq V\) that is independent (for all \(\shortEdge{u}{v}\in\edgeSet\) at most one of \(u\) and \(v\) is in \(S\)) and dominating (for all \(v\in\vertexSet\) at least one \(u\in N^+(v)\) is in \(S\)).
\end{namedthm}

There is no efficient \(n^{1-\varepsilon}\)-approximation for any \(\varepsilon>0\) for the MIDS problem, unless \(\Po=\NP\) \cite[corollary 3]{Halldorsson1993Approximating}.

\begin{proof}[Proof of \cref{thm: complexity vertex stab G w c}]
    The decision variant of the problem asks to find a vertex-stabilizer of size at most \(k\). This problem is in \NP, since if a vertex set \(S\) is given, it can be verified in polynomial time if \(\Abs{S}\leq k\) and if \(\nu^\capacity(\graph\setminus S)=\nu_f^\capacity(\graph\setminus S)\). We prove the \NP-hardness and approximation factor by given an approximation-preserving reduction from the MIDS problem.

    Let \(\graphVE\) be an instance of the MIDS problem. For \(v\in\vertexSet\), we define the gadget \(\Gamma_v\) by
    \begin{gather}
        \vertexSet(\Gamma_v) = N^+(v) \cup \CBra{v_1,v_2,v_3,v_4}, \\
        \edgeSet(\Gamma_v) = \CBra{\shortEdge{u}{v_1} : u\in N^+(v)} \cup \CBra{\shortEdge{v_1}{v_2},\shortEdge{v_2}{v_3},\shortEdge{v_3}{v_4},\shortEdge{v_2}{v_4}}.
    \end{gather}
    For \(e=\shortEdge{u}{v}\in\edgeSet\) and \(i\in\CBra{1,\ldots,n}\), we define the gadget \(\Gamma_{uv}^i\) by
    \begin{gather}
        \vertexSet(\Gamma_{uv}^i) = \CBra{u,v,e_1^i,e_2^i,e_3^i,e_4^i,e_5^i},
        \\
        \edgeSet(\Gamma_{uv}^i) = \CBra{\shortEdge{u}{e_1^i},\shortEdge{v}{e_1^i},\shortEdge{e_1^i}{e_2^i},\shortEdge{e_1^i}{e_3^i},\shortEdge{e_3^i}{e_4^i},\shortEdge{e_4^i}{e_5^i},\shortEdge{e_3^i}{e_5^i}}.
    \end{gather}
    See \cref{fig: complexity proof gadgets} for an example of these gadgets.
    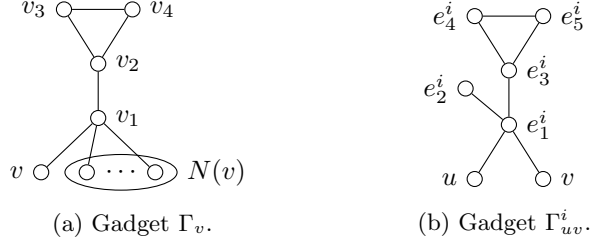
\begin{figure}[t]
        \centering
        \begin{subfigure}{.3\textwidth}
            \centering
            \begin{tikzpicture}[
    scale=.6,
    circ/.style={
        circle,
        draw=black,
        inner sep=2pt
    }
]

\node[circ,label={left :\(v_3\)}] (v3) at (0,0) {};
\node[circ,label={right:\(v_4\)}] (v4) at (1.5,0) {};
\node[circ,label={right:\(v_2\)}] (v2) at (.75,-1.2) {};
\node[circ,label={right:\(v_1\)}] (v1) at (.75,-2.4) {};
\node[circ,label={left:\(v\)}] (v) at (-.5,-3.6) {};
\node[circ] (Nv1) at (.5,-3.6) {};
\node (Nvdots) at (1.25,-3.6) {\(\cdots\)};
\node[circ] (Nv2) at (2,-3.6) {};
\node[ellipse,draw=black,inner sep=-1pt,fit=(Nv1) (Nvdots) (Nv2),label={right:\(N(v)\)}] (Nv) {};

\draw (v3) to (v4);
\draw (v4) to (v2);
\draw (v2) to (v3);
\draw (v2) to (v1);
\draw (v1) to (v);
\draw (v1) to (Nv1);
\draw (v1) to (Nv2);

\end{tikzpicture}
            \caption{Gadget \(\Gamma_v\).}
            \label{fig: complexity proof gadget gamma_v}
        \end{subfigure}
        \begin{subfigure}{.3\textwidth}
            \centering
            \begin{tikzpicture}[
    scale=.6,
    circ/.style={
        circle,
        draw=black,
        inner sep=2pt
    }
]

\node[circ,label={left :\(e_4^i\)}] (s4) at (0,0) {};
\node[circ,label={right:\(e_5^i\)}] (s5) at (1.5,0) {};
\node[circ,label={right:\(e_3^i\)}] (s3) at (.75,-1.2) {};
\node[circ,label={left:\(e_2^i\)}] (s2) at (-.2,-1.6) {};
\node[circ,label={right:\(e_1^i\)}] (s1) at (.75,-2.4) {};
\node[circ,label={left:\(u\)}] (u) at (0,-3.6) {};
\node[circ,label={right:\(v\)}] (v) at (1.5,-3.6) {};

\draw (s4) to (s5);
\draw (s5) to (s3);
\draw (s3) to (s4);
\draw (s3) to (s1);
\draw (s2) to (s1);
\draw (s1) to (u);
\draw (s1) to (v);

\end{tikzpicture}
            \caption{Gadget \(\Gamma_{uv}^i\).}
            \label{fig: complexity proof gadget gamma_uv}
        \end{subfigure}
        \caption{Examples of gadgets.}
        \label{fig: complexity proof gadgets}
    \end{figure} 
    Now let \(\graph'\) be defined as the union of all \(\Gamma_v\) and all \(\Gamma_{uv}^i\), such that vertices from \(\vertexSet\) overlap. We set the capacity as follows: \(\capacity_v=d_v^{\edgeSet(\graph')}\) for all \(v\in\vertexSet\), \(\capacity_{v_1}=d_v^{\edgeSet}+1\) for all \(v\in\vertexSet\), \(\capacity_{e_1^i}=\capacity_{e_3^i}=2\) for \(e_1^i, e_3^i\in\vertexSet(\Gamma_{uv}^i)\) for all \(e=\shortEdge{u}{v}\in\edgeSet\) and \(i\in\CBra{1,\ldots,n}\), and \(\capacity_v=1\) for all remaining \(v\in\vertexSet(\graph')\). All edges are set to have unit-weight. The key point is:

    \begin{claim}
        \(\graph\) has an independent dominating set of size at most \(k\) if and only if \((\graph',\1,\capacity)\) has a vertex-stabilizer of size at most \(k\).
    \end{claim}
    \begin{proof}
        (\(\Rightarrow\))
        Let \(S\) be an independent dominating set of \(\graph\) of size \(k\). The vertices in \(S\) naturally correspond with vertices in \(\graph'\). We show that \(S\) is a vertex-stabilizer of \((\graph',\1,\capacity)\).
        
        We define a \(\capacity\)-matching \(\matching\) and fractional vertex cover \((y,z)\) on \(\graph'\setminus S\) as follows. First, set \(y_v=0\) for all \(v\in\vertexSet\setminus S\). 

        Next, for all \(v\in\vertexSet\), consider \(\Gamma_v\). Add \(\CBra{\shortEdge{u}{v_1} : u\in N^+(v)\setminus S} \cup \CBra{\shortEdge{v_1}{v_2},\shortEdge{v_3}{v_4}}\) to \(\matching\).  Note that at least one vertex from \(N^+(v)\) is in \(S\), since \(S\) is dominating. Set \(y_{v_1}=0\), \(y_{v_2}=1\), \(y_{v_3}=y_{v_4}=0.5\), \(z_e=1\) for all \(e\in\CBra{\shortEdge{u}{v_1} : u\in N^+(v)\setminus S}\) and \(z_e=0\) for the remaining edges in the gadget.
    
        Finally, for all \(e=\shortEdge{u}{v}\in\edgeSet\) and \(i\in\CBra{1,\ldots,n}\), consider \(\Gamma_{uv}^i\). Since \(S\) is dominating, at most one of \(u\) and \(v\) is in \(S\). If neither are in \(S\), add both \(\shortEdge{u}{e_1^i}\) and \(\shortEdge{v}{e_1^i}\) to \(\matching\). If one of them is in \(S\), without loss of generality let it be \(u\), then add \(\shortEdge{v}{e_1^i}\) and \(\shortEdge{e_1^i}{e_2^i}\) to \(\matching\). Furthermore, add \(\shortEdge{e_3^i}{e_4^i}\) and \(\shortEdge{e_3^i}{e_5^i}\) to \(\matching\). Set \(y_{e_1^i}=1\), \(y_{e_2^i}=0\), \(y_{e_3^i}=y_{e_4^i}=y_{e_5^i}=0.5\), and \(z_f=0\) for all edges \(f\) in the gadget.
    
        Let \(x\) be the indicator vector of \(\matching\). One can verify that \(x\) and \((y,z)\) satisfy the complementary slackness conditions for \(\nu^\capacity_f(\graph'\setminus S)\) and \(\tau^\capacity_f(\graph'\setminus S)\). Since \(x\) is integral, this implies that \(\graph'\setminus S\) is stable.

        (\(\Leftarrow\))
        Let \(S\) be a vertex-stabilizer of \((\graph',\1,\capacity)\) of size \(k\). We show that: 
        (i) \(S\) contains at least one vertex of each gadget \(\Gamma_v\); 
        (ii) without loss of generality, one can assume that at most one of \(u\) and \(v\) is in \(S\) for each edge \(\shortEdge{u}{v}\in\edgeSet\).
        
        (i)
        Suppose for the sake of contradiction that there is some \(v\in\vertexSet\) such that \(S\) contains no vertices of \(\Gamma_v\). Since \(\graph'\setminus S\) is stable, there is a maximum-cardinality fractional \(\capacity\)-matching \(x^*\), that is integral. Define for each \(e\in\edgeSet(\graph'\setminus S)\)
        \begin{equation}
            x_e = \begin{cases}
                x^*_e & \text{ if } e\in\edgeSet(\graph'\setminus S)\setminus\edgeSet[\Gamma_v], \\
                1 & \text{ if } e\in\CBra{\shortEdge{u}{v_1} : u\in N^+(v)}, \\
                0 & \text{ if } e=\shortEdge{v_1}{v_2}, \\
                0.5 & \text{ if } e\in\CBra{\shortEdge{v_2}{v_3},\shortEdge{v_3}{v_4},\shortEdge{v_2}{v_4}}.
            \end{cases}
        \end{equation}
        Note that \(x\) is a fractional \(\capacity\)-matching in \(\graph'\setminus S\), since \(x^*\) is. However, \(\sum_{e\in\edgeSet[\Gamma_v]} x_e = d_v+2.5 > \sum_{e\in\edgeSet[\Gamma_u]} x^*_e\), since \(x^*\) is integral. Hence, \(\1^\top x > \1^\top x^*\), contradicting the optimality of \(x^*\).
        
        (ii)
        Suppose there is some \(e=\shortEdge{u}{v}\in\edgeSet\) such that \(S\) contains both \(u\) and \(v\). All gadgets \(\Gamma_{uv}^i\) are then components in \(\graph'\setminus S\). If \(u\) and \(v\) are the only vertices in \(S\) from some component \(\Gamma_{uv}^i\), then a maximum-cardinality fractional \(\capacity\)-matching in this components is given by \(x_{e_1^i e_2^i}=x_{e_1^i e_3^i}=1\) and \(x_{e_3^i e_4^i}=x_{e_4^i e_5^i}=x_{e_3^i e_5^i}=0.5\). Which means this component is not stable, and thus \(\graph'\setminus S\) is not stable, a contradiction. Hence, \(S\) must contain at least one vertex of each \(\Gamma_{uv}^i\) that is neither \(u\) nor \(v\). Consequently, \(k=\Abs{S}\geq n+2\). Since \(\graph\) has only \(n\) vertices, it obviously has an independent dominating set of size at most \(n\), and hence of size at most \(k\). Such a set can for example be obtained by a greedy approach. Hence, for the remainder of the proof we can assume that at most one of \(u\) and \(v\) is in \(S\) for each \(\shortEdge{u}{v}\in\edgeSet\).
        
        We now create a set \(S'\subseteq\vertexSet\) from \(S\), that is an independent dominating set of \(\graph\) of size at most \(k\), as follows.
        Iterate over \(v\in\vertexSet\). Let \(S_v=S\cap\vertexSet(\Gamma_v)\). Note that \(S_v\neq\emptyset\) by (i). Define
        \begin{equation}
            S'_v = \begin{cases}
                \Par{S_v \cup S'} \cap N^+(v) & \text{if this is nonempty}, \\
                v & \text{otherwise}.
            \end{cases}
        \end{equation}
        Set \(S'=S'\cup S'_v\), and repeat for the next vertex.

        Clearly, all \(S'_v\)'s are nonempty, which means that \(S'\) contains at least one vertex from \(N^+(v)\) for all \(v\in\vertexSet\), which means \(S'\) is dominating.

        Suppose for the sake of contradiction that \(S'\) contains both \(u\) and \(v\) for some edge \(\shortEdge{u}{v}\in\edgeSet\). We know \(S\) did not contain both of them, by (ii). If \(S\) contained exactly one of them, without loss of generality let it be \(u\). Then \(\Par{S_v \cup S'} \cap N^+(v)\) also contains \(u\). In particular, this means that we did not add \(v\) to \(S'_v\) and consequently also not to \(S'\), a contradiction. If \(S\) contained neither of them, then because we do the process iteratively, one of them will be added first to \(S'\). Without loss of generality let it be \(u\). Then again \(\Par{S_v \cup S'} \cap N^+(v)\) contains \(u\), so we reach a contradiction in the same way. In conclusion, \(S'\) is independent.

        Before we added \(S'_v\) to \(S'\), we had \(\Abs{S'_v\setminus S'}\leq \Abs{S_v}\). Consequently, \(\Abs{S'}\leq\cup_{v\in\vertexSet} \Abs{S_v} \leq \Abs{S} = k\).
    \end{proof}
    By this claim, any minimum-cardinality vertex-stabilizer of \((\graph',\1,\capacity)\) is of the same size as any minimum independent dominating set of \(\graph\). 
    Further, any efficient \(\alpha\)-approximation algorithm for the vertex-stabilizer problem translates into an efficient \(\alpha\)-approximation algorithm for the MIDS problem. Hence, the result follows from the inapproximability of the MIDS problem.
\end{proof}

\section{Cooperative Matching Games}\label{sec: CMG}

Cooperative matching games in unit-capacity graphs, defined in the introduction, extend quite easily to capacitated graphs, by replacing each \(\nu\) with \(\nu^\capacity\).
In unit-capacity graphs \(\graph\) the following statements are equivalent~\cite{Deng1999Algorithmic,Kleinberg2008Balanced}:
\begin{enumerate}[(i)]
    \item \(\graph\) is stable, 
    \item there exists an allocation in the core of the CMG on \(\graph\),
    \item there exists a stable outcome for the NBG on \(\graph\).
\end{enumerate}
We here note that the equivalence does not extend to capacitated graphs. 

In particular, as mentioned in the introduction, we still have \((i)\iff(iii)\) proven in \cite[corollary 3.3]{Bateni2010Cooperative}. The implication \((i)\implies(ii)\) still holds, and follows from \cite[lemma 3.4]{Bateni2010Cooperative}\footnote{\cite{Bateni2010Cooperative} assumes that the graph is bipartite, but bipartiteness is not needed in their proof.}. However, the graph \(\graph\) given in \cref{fig: counter example} shows that  \((ii)\centernot\implies(i)\) (and hence \((ii)\centernot\implies(iii)\)). 
    
Assuming all the edges of \(\graph\) in \cref{fig: counter example} have unit weight, it is quite easy to see that \(\nucG=3\) and \(\nufcG=3.5\), thus \(\graph\) is not stable. One can check that \(y=(1,1,1,0)\) is in the core.
\begin{figure}[t]
    \centering
    \begin{tikzpicture}[
    scale=.6,
    circ/.style={
        circle,
        draw=black,
        inner sep=2pt
    },
    notMatched/.style={},
    matched/.style={ultra thick},
    fracMatched/.style={dashed,ultra thick}
]

\node[circ,label={left :\(2\)}] (1) at (0, 0) {};
\node[circ,label={right:\(2\)}] (2) at (2, 0) {};
\node[circ,label={left :\(2\)}] (3) at (0,-2) {};
\node[circ,label={right:\(1\)}] (4) at (2,-2) {};

\draw[matched] (1) to (2);
\draw[matched] (1) to (3);
\draw[matched] (2) to (3);
\draw[notMatched] (2) to (4);
\draw[notMatched] (3) to (4);
\end{tikzpicture}
\hspace{1em}
\begin{tikzpicture}[
    scale=.6,
    circ/.style={
        circle,
        draw=black,
        inner sep=2pt
    },
    notMatched/.style={},
    matched/.style={ultra thick},
    fracMatched/.style={dashed,very thick}
]

\node[circ,label={left :\(1\)}] (1) at (0, 0) {};
\node[circ,label={right:\(1\)}] (2) at (2, 0) {};
\node[circ,label={left :\(1\)}] (3) at (0,-2) {};
\node[circ,label={right:\(0\)}] (4) at (2,-2) {};

\draw[matched] (1) to (2);
\draw[matched] (1) to (3);
\draw[fracMatched] (2) to (3);
\draw[fracMatched] (2) to (4);
\draw[fracMatched] (3) to (4);
\end{tikzpicture}
    \caption{On the left: the graph \(\graph\) where the values close to the vertices indicate the capacities. Bold edges indicate a maximum \(\capacity\)-matching. On the right: the graph \(\graph\) where the values close to the vertices indicate the allocation \(y\). A maximum fractional \(\capacity\)-matching is given by \(x_e=\frac{1}{2}\) for dashed edges, \(x_e=1\) otherwise.}
    \label{fig: counter example}
\end{figure}
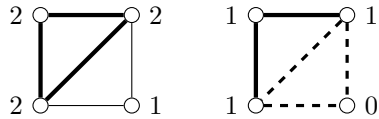

\bibliographystyle{plain}
\bibliography{bib}

\end{document}